\documentclass[11pt]{article}

\def \titlevar{Parity vs. AC\textsuperscript{0} with simple quantum preprocessing}

\usepackage[english]{babel}
\usepackage{bm, amsmath, amsthm, amssymb, mathtools, float, titling, graphicx, hyperref, color, stackengine, braket, wasysym, pdflscape, geometry, bbm, color,authblk}
\stackMath
\usepackage{enumitem}
\setlist{parsep=0pt,listparindent=\parindent}
\usepackage{titling}
\usepackage{ wasysym }

\renewcommand{\epsilon}{\varepsilon}

\newtheorem{theorem}{Theorem}
\newtheorem*{theorem*}{Theorem}
\newtheorem{lemma}[theorem]{Lemma}
\newtheorem{corollary}[theorem]{Corollary}
\newtheorem{defn}{Definition}
\newtheorem*{prop*}{Proposition}
\newtheorem{conj}{Conjecture}
\newtheorem*{conj*}{Conjecture}

\newtheorem*{fact*}{Fact}
\newtheorem{prop}{Proposition}
\newtheorem*{ex*}{Example}
\newtheorem{ex}{Example}

\newtheorem{question}{Question}
\theoremstyle{remark}
\newtheorem{claim}{Claim}

\DeclareMathOperator{\poly}{poly}

\DeclareMathOperator{\tr}{tr}
\DeclareMathOperator*{\E}{\mathbb{E}}

\newcommand*{\bigchi}{\raisebox{.1em}{\scalebox{1.2}{$\chi$}}}
\newcommand*{\medchi}{\raisebox{.1em}{\scalebox{1}{$\chi$}}}

\newcommand{\R}{\mathbb{R}}

\newcommand{\mc}[1]{\mathcal{#1}}
\newcommand{\bs}[1]{\boldsymbol{#1}}

\newcommand{\restr}[1]{{\upharpoonright_{#1}}}

\newcommand{\0}{\textsuperscript{0}}

\newcommand{\ketbra}[2]{{\ket{#1}\!\!\bra{#2}}}

\newcommand{\ac}{\mathsf{AC^0}}
\newcommand{\qnc}{\mathsf{QNC^0}}
\newcommand{\qac}{\mathsf{QAC^0}}
\newcommand{\nc}{\mathsf{NC^0}}
\newcommand{\acqnc}{{\ac\!\bs\circ\qnc}}

\newcommand{\parity}{\text{\textsc{Par}}}

\title{\titlevar\vspace{-1ex}}
\author{Joseph Slote\footnote{California Institute of Technology. Email: \url{jslote@caltech.edu}.}}
\date{}

\begin{document}
\maketitle
\begin{abstract}
A recent line of work \cite{bravyi_quantum_2018, watts_exponential_2019, 
grier_interactive_2020, bravyi_quantum_2020,  watts_unconditional_2023} has shown the unconditional advantage of constant-depth quantum computation, or $\qnc$, over $\nc$, $\ac$, and related models of classical computation.
Problems exhibiting this advantage include search and sampling tasks related to the parity function, and it is natural to ask whether $\qnc$ can be used to help compute parity itself.
Namely, we study $\acqnc$---a hybrid circuit model where $\ac$ operates on measurement outcomes of a $\qnc$ circuit---and we ask whether \textsc{Par} $\in\acqnc$.

We believe the answer is negative.
In fact, we conjecture $\acqnc$ cannot even achieve $\Omega(1)$ correlation with parity.
As evidence for this conjecture, we prove:

\begin{itemize}
	\item When the $\qnc$ circuit is ancilla-free, this model can achieve only negligible correlation with parity, even when $\ac$ is replaced with any function having LMN-like decay in its Fourier spectrum.
	\item For the general (non-ancilla-free) case, we show via a connection to nonlocal games that the conjecture holds for any class of postprocessing functions that has approximate degree $o(n)$ and is closed under restrictions.
		Moreover, this is true even when the $\qnc$ circuit is given arbitrary quantum advice.
		By known results \cite{bun_quantum_2019}, this confirms the conjecture for linear-size $\ac$ circuits.
		
	\item Another approach to proving the conjecture is to show a switching lemma for $\acqnc$.
		Towards this goal, we study the effect of quantum preprocessing on the decision tree complexity of Boolean functions.
		We find that from the point of view of decision tree complexity, nonlocal channels are no better than randomness: a Boolean function $f$ precomposed with an $n$-party nonlocal channel is together \emph{equal} to a randomized decision tree with worst-case depth at most $\mathrm{DT}_\mathrm{depth}[f]$.
		\end{itemize}

Taken together, our results suggest that while $\qnc$ is surprisingly powerful for search and sampling tasks, that power is ``locked away'' in the global correlations of its output, inaccessible to simple classical computation for solving decision problems.
\\

\footnotesize
\noindent Keywords: $\qnc$, $\ac$, nonlocal games, $k$-wise indistinguishability, approximate degree, switching lemma, Fourier concentration
\end{abstract}
\newpage

\section*{Introduction}
In 2017, Bravyi, Gosset, and K\"onig \cite{bravyi_quantum_2018} proved a breakthrough unconditional separation between constant-depth quantum circuits, or $\qnc$, and constant-depth bounded fan-in classical circuits, or $\nc$.
The authors showed that for a certain search problem solvable by $\qnc$ circuits, any randomized $\nc$ circuit solving the same problem with high probability must have logarithmic depth.
The realization that unconditional proofs of quantum advantage were possible---albeit over weak models of classical computation---inspired an exciting series of results strengthening and generalizing the work of Bravyi, Gosset, and K\"onig.
There are now separations against stronger classical circuit models such as constant depth circuits with unbounded fan-in, or $\ac$ \cite{watts_exponential_2019}, average-case separations \cite{gall_average-case_2020}, separations between more intricate interactive models \cite{grier_interactive_2020}, separations that remain even for quantum circuits subject to noise (\emph{e.g.,} \cite{bravyi_quantum_2020}), and separations for sampling problems with no input \cite{watts_unconditional_2023}, among others.

Although these separations are for comparatively weak models of computation, they are concrete non-oracle, non-query separations, and are free from complexity-theoretic assumptions, making them important companions to the query complexity and conditional separations studied since the founding of quantum computer science.
One notable feature of these $\qnc$ separations, however, is that they are all for search or sampling problems; decision separations appear to be absent from this list.

On the surface, there is a somewhat trivial reason for this: $\qnc$ cannot solve interesting decision problems alone.
Indeed, any single output qubit in a constant-depth quantum circuit can only depend on constantly-many input qubits, so any $\qnc$ circuit with one output bit may be simulated by randomized $\nc$.
However, this ``lightcone barrier'' may be removed by instead measuring all qubits in the quantum circuit and then applying a classical Boolean function $f$ to the result.
As long as $f$ depends on all of its inputs, it might be possible for $f$ to leverage $\qnc$'s search and sampling prowess for decision-making ends.
Given Bene Watts et al.'s search separation between $\qnc$ and $\ac$ \cite{watts_exponential_2019}, a natural class of Boolean functions to choose for this postprocessing is $\ac$ itself.
This gives rise to the following definition, which does not appear to have been studied before.

\begin{defn}
	Let $\acqnc$ denote the model of computation composed of a $\qnc$ circuit
	$\mc C$, followed by a computational basis measurement, and then an $\ac$ function $f$ applied to the result.
	This process defines the randomized Boolean function $f\circ\mc C:\{0,1\}^n\to\mc M(\{-1,1\})$ from the hypercube to the set $\mc M(\{-1,1\})$ of probability measures  on $\{-1,1\}$.
	\end{defn}
In this work we take a $\qnc$ circuit to be a polynomial-size constant-depth quantum circuit composed of arbitrary 2-qubit unitary gates.
	Ancilla qubits are allowed and are initialized in the state $\ket{0^m}$ for $m\in\poly(n)$.
	No geometric locality or clean computation constraints are assumed.
	A formal definition appears later as Definition \ref{def:qnc}.

Certainly $\qnc\subseteq\acqnc$, so the search separation between $\qnc$ and $\ac$ in Bene Watts et al. is also a search separation between $\acqnc$ and $\ac$.
Moreover, this modification obviates the lightcone barrier mentioned above and allows us to ask meaningful questions about decision separations between concrete models of quantum and classical computation.

Specifically, Bene Watts et al. \cite{watts_exponential_2019} show exponential advantage of $\qnc$ over $\ac$ for (a variant of) the ``parity halving problem'': 
\begin{quote}
\emph{Parity halving.} Given $x\in \{0,1\}^n$ with the promise $|x|\equiv 0\mod 2$, output any even string if $|x|\equiv 0\mod 4$ and any odd string otherwise.
\end{quote}
Given the form of this problem, it is natural to ask whether parity is itself computable by a hybrid model such as $\acqnc$.

\medskip
Before summarizing our progress on this question, we pause to note another reason to study $\acqnc$ coming from the rich subject of quantum-classical interactive proofs.
A central project in this area is the classical verification of quantum computations \cite{Gheorghiu2018}.
In a landmark 2018 work, Mahadev gave a cryptographic protocol for this task \cite{mahadev}; however, whether or not this task may be accomplished without cryptographic hardness assumptions remains open despite many efforts \cite{Gheorghiu2018}.
It therefore makes sense to consider the question in simpler contexts, such as where the prover and verifier are replaced with $\qnc$ and $\ac$ respectively and interact for constantly-many rounds to establish the correctness of a $\qnc$ computation.
With this perspective we see that $\acqnc$ models the first round of interaction in such a proof system.

\subsection*{Parity vs. $\acqnc$: Overview and organization}

We conjecture that $\acqnc$ cannot approximate parity ($\parity_n$) on average, over both choice of uniformly random input $x\sim\mc U(\{0,1\}^n)$ and the randomness in $f\circ \mc C$.
It is convenient to take $\parity$ and $f\circ\mc C$ to be $(\pm1)$-valued and phrase this in terms of the correlation
\[\E_x[(f\circ\mc C)(x)\cdot \parity(x)],\]
proportional to the advantage of $f\circ \mc C$ over random guessing for computing parity.
\begin{conj}
\label{conj:main}
	$\acqnc$ cannot achieve correlation $\Omega(1)$ with the parity function.
	That is, fix a polynomial size bound $p(n)$ and constant depth $d$.
	Then for all sequences $\{(f_n,\mc C_n)\}_n$ of circuits such that $\mathrm{size}(f_n),\mathrm{size}(\mc C_n)\leq p(n)$ and $\mathrm{depth}(f_n),\mathrm{depth}(\mc C_n)\leq d$, we have
	\[\E_x[(f_n\circ\mc C_n)(x)\cdot \parity_n(x)]\to 0 \quad\text{as}\quad n\to \infty\,.\]
\end{conj}

Although proving correlation bounds against $\ac$ is a well-understood topic with many techniques (among them H\aa stad's switching lemma \cite{hastad_almost_1986} and Razborov-Smolensky \cite{razborov_lower_1987,smolensky_algebraic_1987}), when $\qnc$ precomputation is added these approaches cannot be used directly.
The pursuit of new techniques leads us to connections with many-player nonlocal games, approximate degree bounds, and new directions for generalizing H\aa stad's switching lemma.
Evidence for Conjecture \ref{conj:main} is laid out as follows.

\subsubsection*{The ancilla-free case}
In Section \ref{sec:ancilla-free} we prove Conjecture \ref{conj:main} when $\qnc$ is restricted to be ancilla-free.
A key feature of such $\qnc$ circuits is that they correspond to unitary transformations, and we find in this case the correlation of $f\circ\mc C$ with $\parity$ is controlled by the Fourier tail of $f$.
Recall the $k$\textsuperscript{th} \emph{Fourier tail} of a Boolean function $f$ is given by
\[\mathbf{W}^{\geq k}[f]:=\textstyle\sum_{|S|\geq k}\widehat{f}(S)^2\,.\]
Appealing to the Linial-Mansour-Nisan-type (LMN-type) estimates of the Fourier tail of $\ac$ \cite{linial_constant_1993}, we obtain the following strong correlation bound.

\begin{theorem}[Ancilla-free $\qnc$, general $\ac$ case]
	\label{thm:ancilla-free-intro}
	If $\mc C$ is an ancilla-free $\qnc$ circuit and $f$ is an $\ac$ function then
	\[\E_x[(f\circ\mc C)(x)\cdot \parity_n(x)]\leq 2^{-n/\mathrm{polylog}(n)}\,.\]
\end{theorem}
\noindent This is proved as Corollary \ref{cor:afq-gac} in Section \ref{sec:ancilla-free}.
The full statement holds for any Boolean function $f$ with sufficient decay in the tail of the Fourier spectrum, including those outside of $\ac$.

However, as we explain in the end of Section \ref{sec:ancilla-free}, the proof technique of Theorem \ref{thm:ancilla-free-intro} cannot extend to the case of general $\qnc$ and we must find a different approach.

\subsubsection*{Reducing to nonlocal games}

To move beyond ancilla-free $\qnc$, in Section \ref{sec:nlgs} we reduce Conjecture \ref{conj:main} to a question about the value of a certain class of nonlocal games, which we call \emph{$n$-player parity games} and which are parameterized by a postprocessing Boolean function $f$.
Through a connection to the notion of $k$-wise indistinguishability introduced in \cite{bogdanov_bounded_2016}, we show the quantum value of a parity game is controlled by the approximate degree of the associated $f$.

Recall for $\epsilon>0$ the \emph{$\epsilon$-approximate degree} of a $(0,1)$-valued\footnote{For $(\pm 1)$-valued $f$, we use the same definition after making the standard identification $+1\mapsto 0, -1\mapsto 1$.
} Boolean function $f$ is given by
\[\widetilde{\text{deg}}_\epsilon[f]=\min\{\deg(g)\mid g:\{0,1\}^n\to\mathbb{R}\text{ a polynomial with } \|f-g\|_\infty\leq \epsilon\}\,.\]
Of course, $\widetilde{\deg}_\epsilon[f]\leq n$ for any $n$-variate $f$ and $\epsilon >0$.
By convention $\widetilde{\deg}[f]:=\widetilde{\deg}_{1/3}[f]$.
A \emph{function class} $\mc F=(\mc F_n)_{n\geq 1}$ is a sequence of sets $\mc F_n$ of $n$-variate Boolean functions, and we extend approximate degree to function classes via $\widetilde{\deg}[\mc F](n) := \max_{f\in\mc F_n}\widetilde{\deg}[f]$.
With this notation, we have the following theorem.

\begin{theorem}[Corollary \ref{cor:deg}, Section \ref{sec:nlgs}]
	\label{thm:cordeg}
	Suppose function class $\mc F$ is closed under inverse-polynomial-sized restrictions.
		Then if $\widetilde{\deg}[\mc F]\in o(n)$, $\mc F\circ\qnc$ cannot achieve $\Omega(1)$ correlation with $\parity_n$, even if $\qnc$ is given arbitrary quantum advice.
\end{theorem}

It follows from Theorem \ref{thm:cordeg} that Conjecture 1 would be confirmed in full generality if $\widetilde{\text{deg}}[\ac]\in o(n),$ a notorious open problem \cite{bun_approximate_2022}.
Such a bound is already known for large subclasses of $\ac$, however: for example, for $\ac$ circuits of size $\mathcal{O}(n)$ (termed $\mathsf{LC^0}$), we may appeal to the recent bounds of \cite{bun_quantum_2019} to conclude:

\begin{theorem}[General $\qnc$, linear-size $\ac$ case]
	Suppose $f\in \ac$ has size $\mathcal{O}(n)$.
	Then $f\circ\qnc$ achieves correlation at most $1/\poly(n)$ with $\parity_n$.
	This holds even if $\qnc$ is given arbitrary quantum advice.
	That is,
	\[\E[(\mathsf{LC^0\circ QNC^0/{qpoly}})\cdot\parity_n]\in \textnormal{negl}(n)\,.\]
	\end{theorem}
\noindent (This is proved as Corollary \ref{cor:qnc0lin} in Section \ref{sec:nlgs}).

Is the difficulty of proving approximate degree bounds for $\ac$ a barrier for resolving Conjecture \ref{conj:main}?
It seems unlikely: the reduction to approximate degree bounds is via a series of substantial relaxations and it would be surprising if all the required converses held.
In fact, we conclude Section \ref{sec:nlgs} with a self-contained approximation theory question (Question \ref{q:blockwise}) concerning a notion of blockwise approximate degree which may be easier to solve than $\widetilde{\text{deg}}[\ac]$ but would still imply Conjecture \ref{conj:main}.

\subsubsection*{Towards an $\acqnc$ switching lemma}
In Section \ref{sec:switching-lemma} we chart a different route to resolving Conjecture \ref{conj:main}, aiming to prove a switching lemma for our hybrid $\acqnc$ circuits.
Recall that H\aa stad's original switching lemma is used to argue that (very roughly) randomly fixing a large fraction of inputs to an $\ac$ circuit with high probability yields a function that can be computed by a shallow decision tree.
At the same time, $\parity$ retains maximum decision tree complexity under the same restrictions, so this leads to $\ac$ correlation bounds.

In comparison to H\aa stad's switching lemma and its descendants, a challenge with $\acqnc$ circuits is that $\qnc$ can correlate, spread out, and bias random restrictions before they reach the bottom layer of DNFs or CNFs in the $\ac$ circuit.
If $\qnc$ were replaced with randomized $\nc$ this problem could be readily addressed by considering each deterministic circuit in the distribution, applying standard arguments there, and computing the expected correlation with parity across circuits in the distribution.
But unlike randomized computation, and as discussed \emph{e.g.,} in \cite{aaronson_acrobatics_2022}, a recurring theme in quantum complexity theory is the impossibility of ``pulling out the quantumness'' from a quantum circuit.

Contrary to this theme, however, we show that when $\qnc$ is replaced by an $n$-party nonlocal channel $\mc N$, it \emph{is} possible to pull out the quantumness in a particular sense:

\begin{theorem*}[Theorem \ref{thm:tree-decomp}, restated]
Let $f:\{0,1\}^{m}\to\{0,1\}$ be any Boolean function and consider an $n$-party nonlocal channel $\mc N$, where the $i$\textsuperscript{th} party receives one bit and responds with $m_i\geq0$ bits, such that $\sum_im_i=m$.
Then the random function $f\circ \mc N$ is equal to a randomized decision tree $\Gamma$ such that $\mathrm{depth}(T)\leq \mathrm{DT_{depth}}[f]$ for all $T\in\mathrm{Supp}(\Gamma)$.
\end{theorem*}

\noindent(This theorem is proved in Section \ref{sec:switching-lemma} as Theorem \ref{thm:tree-decomp}.)
By an \emph{$n$-party nonlocal channel} we mean the channel corresponding to a quantum strategy in an $n$-player nonlocal game: parties receive one bit of input each and may measure disjoint systems of a shared quantum state as part of their responses, but they are not allowed to communicate.
A formal definition appears as Definition \ref{def:nonlocal}.
In fact, Theorem \ref{thm:tree-decomp} is true not only for nonlocal channels, but for any channel where parties obey the no-signaling property; that is, the output of any subset $S\subset[n]$ of the parties is a function only of the inputs to those parties in $S$.
A formal definition of no-signaling channels appears as Definition \ref{defn:no-signaling}.

The regime where Theorem \ref{thm:tree-decomp} is truly interesting is when $\mathrm{DT_{depth}}[f]\geq \log(n)$.
Then $f$ may depend on all the input coordinates and (potentially) make great use of the processing power afforded by no-signaling channels.
Theorem \ref{thm:tree-decomp} says that to the contrary, precomposition of $f$ by any no-signaling channel has no effect on the (randomized) decision tree complexity of $f$.

How does Theorem \ref{thm:tree-decomp} connect to Conjecture \ref{conj:main}?
As we detail in Section \ref{sec:nlgs}, the replacement of $\qnc$ by the channel $\mc N$ is essentially without loss of generality from the point of view of Conjecture \ref{conj:main}.
Unfortunately, however, $\ac$ circuits can easily have maximum decision tree complexity, so Theorem \ref{thm:tree-decomp} cannot be immediately applied.
Instead, we believe Theorem \ref{thm:tree-decomp} stands as a striking example of the inability of classical postprocessing to make use of the search and sampling power of quantum and super-quantum models of computation.
Additionally, we hope that this theorem's proof technique, which involves tracking the interplay between a decision tree for $f$ and the no-signaling channel $\mc N$, represents the style of argument that could eventually lead to a switching lemma for $\acqnc$.

\subsection*{Outlook}
Taken together, these results suggest $\qnc$ cannot render its power in a way $\ac$ or other simple models of classical computation can access for the purpose of making decisions.
Several questions for further research are posed in Section \ref{sec:discussion}.

\subsection*{Related work}
Unlike the quantum-classical separations surveyed in the introduction, which show quantum upper bounds and classical lower bounds, this paper aims to prove a lower bound against a concrete model of quantum computation.
The pursuit of lower bounds against quantum circuits for computational problems is a nascent area and very little is known.

One quantum circuit model where lower bounds have received some concerted study is $\mathsf{QAC^0}$ \cite{moore_quantum_1999, hoyer_quantum_2003,pade_depth-2_2020, rosenthal_bounds_2021, paulispec}.
A superset of $\qnc$ circuits, $\qac$ additionally allows for arbitrarily-large Toffoli gates,
\[\ket{x_1,\ldots, x_k, x_{k+1}}\mapsto\big|x_1,\ldots,x_k,\,x_{k+1}\oplus (\wedge_{i=1}^k x_i)\big\rangle\,,\]
which are quantum analogues of classical AND gates with unbounded fan-in.
In this setting correlation with parity is also a central open question, and there is growing evidence that $\qac$ cannot achieve $\Omega(1)$ correlation with parity either.
Recent work has shown negligible correlation bounds between $\mathsf{QAC^0}$ and parity when a) the $\mathsf{QAC^0}$ circuit is restricted to depth $2$ \cite{rosenthal_bounds_2021}, and b) when the $\qac$ circuit is of any depth $d$ and is restricted to $\mathcal{O}(n^{1/d})$-many ancillas \cite{paulispec}.
In fact, the second result is a corollary to a Pauli-basis analogue of the LMN theorem for the same subclass of $\qac$  \cite{paulispec}.

The relationship between $\mathsf{QAC^0}$ and $\acqnc$ is rather unclear, and they are likely incomparable as decision classes.
In fact, as far as we know, it is even open whether $\ac\subseteq \qac$, let alone whether $\acqnc\subseteq\qac$ (noting the trivial containment $\ac\subseteq\acqnc$).

The difficulty in comparing these models stems from a subtlety concerning the difference between unbounded fan-in and unbounded fan-out when implemented coherently.
$\ac$ circuits have no restriction on the \emph{fan-out} of their gates, while the definition of $\qac$ appears to strongly limit outward propagation of information.
If one augments $\qac$ with the so-called \emph{fan-out gate}---which is a CNOT gate with any number of target qubits,
\[\ket{x_1,\ldots, x_k}\mapsto\ket{x_1,x_1\oplus x_2,\ldots, x_1\oplus x_k}\,,\]
one obtains the circuit model $\mathsf{QAC}^0_f$, and it is known $\mathsf{QAC}^0_f$ can compute parity exactly in depth 3 \cite{moore_quantum_1999}.
In view of existing lower bounds against $\qac$, it is expected that $\qac$ is strictly contained in $\mathsf{QAC}^\mathsf{0}_f$, and assuming this holds we immediately have that the function version of $\acqnc$ is not in the function version of $\qac$.
This follows, for example, from the fact that multi-output $\ac$ circuits easily implement the classical reversible fan-out gate, $(x_1,\ldots, x_k)\mapsto (x_1,x_1\oplus x_2,\ldots, x_1\oplus x_k)$.
It is safe to say the interaction of nonlocal gates with $\qnc$---whether that interaction is coherent as in $\qac$ and $\mathsf{QAC}^\mathsf{0}_f$, or preceded by measurement as in $\acqnc$---is only beginning to be understood.

\medskip

A separate area where concrete quantum circuit lower bounds have been very successfully developed is for \emph{state preparation} problems.
We do not attempt a survey here, but just mention they were crucial to the resolution of the NLTS conjecture \cite{anshu_nlts_2023} and make use of ideas from error correction, which partially originate in sampling lower bounds from classical complexity \cite{lovett_bounded-depth_2011}.
However, it is not clear how to transfer these methods to quantum circuit lower bounds for computational problems in the $\acqnc$ model.

\section{Lower bounds when $\qnc$ is ancilla-free}
\label{sec:ancilla-free}
Here we show any Boolean function $f$ with small Fourier tail retains a small top-degree coefficient when composed with ancilla-free $\qnc$.
By the celebrated work of H\aa stad \cite{hastad_almost_1986} and Linial, Mansour, and Nisan \cite{linial_constant_1993}, any $f\in\ac$ is an example---but this theorem addresses a broader set of functions.
On the other hand, as we discuss at the end of the section, once ancillas are allowed, the theorem no longer holds for such a general class of functions.

Recall a function $f:\{-1,1\}^n\to \R$ admits a unique \emph{Fourier decomposition}
\[f=\sum_{S\subseteq[n]}\widehat{f}(S)\bigchi_S,\]
where $\bigchi_S(x):=\prod_{i\in S}x_i$ is the $S$\textsuperscript{th} Fourier character (see \emph{e.g.,} \cite{ODonnell2014-li} for more).
We will later make use of the familiar Plancherel theorem, which states for any $f,g:\{-1,1\}^n\to\mathbb{R}$ that
\[\E_x[f(x)g(x)]=\sum_{S\subseteq[n]}\widehat{f}(S)\widehat{g}(S)\,.\]

Let us briefly connect this perspective to quantum observables.
Given a Boolean function $f:\{\pm 1\}^n\to \R$ we define its \emph{Von Neumann observable} as
\[M_f:=\sum_xf(x)\ket{x}\!\bra{x}\,.\]
An identity we will use is
\[M_{\medchi_S}=Z^S\,,\]
where the operator $Z$ here is the Pauli operator $\left(\begin{smallmatrix}1&0\\0&-1\end{smallmatrix}\right)$, and generally for any 1-qubit operator $A$ we use the notation
\[A^S := \bigotimes_i\begin{cases}
    A &\text{if } i\in S\\
    \mathbbm{1} &\text{otherwise.}
\end{cases}
\]

Any Von Neumann observable $M$ (that is, any Hermitian operator) has expectation value on state $\rho$ given by
\[\langle M\rangle_\rho:=\tr[M \rho]\,,\]
and when $M=M_f$ and $x\in\{0,1\}^n$ we note the identity \[\langle M_f\rangle_x:=\langle M_f\rangle_{\ket{x}\!\bra{x}} = f(x)\,.\]

With this notation, we prove the following.

\begin{theorem}[Correlation bound for ancilla-free $\qnc$]
\label{thm:ancilla-free}
    Let $f:\{\pm1\}^n\to\R$ and $U$ an ancilla-free $\qnc$ circuit of depth $t$.
    Then the correlation of $f\circ U$ and $\text{\textsc{Par}}$ is bounded as
    \[\E_x[\langle U^\dagger M_fU\rangle_x\cdot \parity_n(x)]\;\leq\; \left(\mathbf{W}^{\geq 2^{-t}n}[f]\right)^{1/2}\,.\]
\end{theorem}

For example, when $f$ is an $\ac$ circuit, we may use an LMN-type Fourier concentration bound, such as from \cite{tal_tight_2017}, to a obtain:

\begin{corollary}
\label{cor:afq-gac}
    If $U$ is an ancilla-free $\qnc$ $n$-qubit circuit of depth $t$, and $f:\{\pm 1\}^n\to \{\pm 1\}$ is implemented by an $\ac$ circuit of depth $d$ and size $s$, we have
    \[\E_x[\langle U^\dagger M_f U\rangle_x\cdot \parity_n(x)]\leq \sqrt{2}\cdot \exp\left({\frac{-n}{2^{t+1}\mathcal{O}(\log s)^{d-1}}}\right).\]
\end{corollary}

The proof of Theorem \ref{thm:ancilla-free} relies on two brief lemmas.
The first says that when measuring correlations, we could just as well have compared the correlation of $f$ alone to the random function $\parity_n\circ U^\dagger$, defined by applying $\text{\textsc{Par}}_n$ to the output of $U^\dagger\ket{x}$.

\begin{lemma}[Symmetry of correlation]
    \label{lem:correlation}
    Let $f,g:\{\pm 1\}^n\to \{\pm1\}$ and $U$ any $n$-qubit unitary.
    Then
    \begin{align*}
    \E_x[\langle U^\dagger M_f U\rangle_x\cdot g(x)]&=\E_x[f(x)\cdot\langle U M_g U^\dagger\rangle_x]\\
    &=2^{-n}\tr[M_f U M_g U^\dagger]\,.
    \end{align*}
\end{lemma}
\begin{proof}
    Expanding the trace we have
    \begin{align}
        \tr[M_fU M_g U^\dagger]&=\sum_z\bra{z}\textstyle\big(\sum_yf(y)\ketbra{y}{y}\big)U\big(\sum_xg(x)\ketbra{x}{x}\big)U^\dagger\ket{z}\nonumber\\
        \label{eq:sym-a}&=\sum_{x,y,z}f(y)g(x)\braket{z|y}\bra{y}U^\dagger\ket{x}\bra{x}U\ket{z}\nonumber\\
        &=\sum_{x,y}f(y)g(x)\bra{y}U\ket{x}\bra{x}U^\dagger\ket{y},
        \end{align}
        while expanding the expectations we see
        \begin{align*}
        \E_x[\langle U^\dagger M_fU\rangle_x\cdot g(x)]= \frac{1}{2^n}\sum_{x,y}f(y)g(x)\bra{x}U^\dagger\ket{y}\bra{y}U\ket{x}=\E_y[f(y)\cdot\langle U M_g U^\dagger \rangle_y]\,.\nonumber
    \end{align*}
    Identifying the center expression with (a multiple of) \eqref{eq:sym-a} and changing variables completes the lemma.
\end{proof}

The second lemma roughly says when Fourier characters $Z_S$ and $Z_T$ correspond to sets $S,T$ of very different cardinality, they remain orthogonal (with respect to the inner product $\langle A,B\rangle = \tr[A^\dagger B]$) after an application of $U$.
\begin{lemma}[Lightcone lemma]
\label{lem:lightcone}
Suppose $U$ is a depth-$t$ ancilla-free quantum circuit and $|S|2^t<n$.
Then
    \[\tr[Z_{[n]}UZ_SU^\dagger] = 0\]
\end{lemma}
\begin{proof}
    The number of qubits on which $Z_S$ acts nontrivially at most doubles upon conjugation by each layer in $U$.
    Therefore the number of non-identity coordinates in $UZ_SU^\dagger$ is at most $|S|2^t$.
    Now if $|S|2^t<n$, then there is at least one coordinate $j$ such that $UZ_SU^\dagger = V_{[n]\backslash j}\otimes \mathbbm{1}_j$ for some $(n-1)$-qubit unitary $V_{[n]\backslash j}$,
    so
    \[\tr[Z_{[n]}UZ_SU^\dagger] = \tr[Z_{[n]}(V_{[n]\backslash j}\otimes \mathbbm{1}_j)]= \tr[Z_{[n]\backslash j}V_{[n]\backslash j}]\cdot\tr[Z]=0\]
    because $Z$ is traceless.
\end{proof}

With these lemmas in hand, we can give the proof of Theorem 1 in a single display:
\begin{proof}[Proof of Theorem \ref{thm:ancilla-free}]
\begin{align*}
    \E_x[\langle U^\dagger M_fU\rangle_x\cdot \bigchi_{[n]}(x)] &= \E_x[f(x)\cdot \langle U Z_{[n]}U^\dagger\rangle_x]\tag{Lemma \ref{lem:correlation}}\\
    &= \sum_{S\subseteq[n]}\widehat{f}(S)\cdot\widehat{\langle U Z_{[n]}U^\dagger\rangle}(S)\tag{Plancherel}\\
    &= \sum_{S\subseteq[n]}\widehat{f}(S)\underbrace{\E_x[\langle U Z_{[n]}U^\dagger\rangle_x\cdot \bigchi_S(x)]}\\
    &\hspace{6em}=2^{-n}\tr[Z_{[n]}U^\dagger Z_{S}U]\tag{Lemma \ref{lem:correlation}}\\
    &\hspace{6em}=0 \quad\text{if } |S|2^t<n\tag{Lemma \ref{lem:lightcone}}\\
    &=\sum_{\substack{S\subseteq[n]\\|S|\geq 2^{-t}n}}\widehat{f}(S)\cdot\widehat{\langle U^\dagger Z_{[n]}U\rangle}(S)\\
    &\leq \left(\sum_{|S|\geq 2^{-t}n}\widehat{f}(S)^2\right)^{1/2}\left(\sum_{|S|\geq 2^{-t}n}\widehat{\langle U^\dagger Z_{[n]}U\rangle}(S)^2\right)^{1/2}\tag{Cauchy-Schwarz}\\
    &\leq \left(\mathbf{W}^{\geq 2^{-t}n}[f]\right)^{1/2}\,.\qedhere
\end{align*}
\end{proof}

One may ask whether this proof approach extends to $\qnc$ circuits with ancillas.
Although it might be possible to prove slight generalizations, we present an example demonstrating that any proof approach using an LMN-type theorem as a black box will fail for general $\qnc$ circuits.
This is essentially because functions with Fourier decay are not closed under composition.

\begin{ex}
Consider the following ``Trojan horse'' function on an even number of bits $n=2m$:
\begin{align*}
    h:\{\pm 1\}^{2m}&\to\{\pm 1\}\\
    x\hspace{1.4em}&\mapsto\begin{cases}
        \bigchi_{[m]}(x) &\text{if } x_{[m+1,2m]}=11\cdots 1\\
        1 &\text{otherwise}\,.
    \end{cases}
\end{align*}
By direct computation one finds the Fourier coefficients of $h$ are given by
\[\widehat{h}(S)=\begin{cases}
	1-2^{-m} & S=\emptyset,\\
	-2^{-m} & S\subseteq[m],S\neq \emptyset\\
	2^{-m} & [m+1,2m]\subseteq S\\
	0 &\text{otherwise}\,.
\end{cases}\]
This means for any $t\geq 1$, the $t$\textsuperscript{th} Fourier tail of $h$ is $\mathbf{W}^{\geq t}[h]\in \mathcal{O}(2^{-n/2})$.
Thus by Theorem \ref{thm:ancilla-free}, for any ancilla-free $\qnc$ circuit $\mc C$, $h\circ \mc C$ has negligible correlation with parity.

On the other hand, consider the (deterministic) function $C:\{\pm 1\}^m\to\{\pm1\}^{2m}$ given by $x\mapsto x11\cdots 1$.
Certainly $C$ can be implemented in $\qnc$, and we have $h\circ C = \bigchi_{[m]} = \parity_m$.
\end{ex}

This example shows that exponential Fourier decay of $f$ is not sufficient to entail Conjecture \ref{conj:main} for general $\acqnc$ circuits.
We must take a different approach that exploits finer structural properties of $\ac$ and $\qnc$.

\section{Lower bounds against $\acqnc$ via nonlocal games}
\label{sec:nlgs}

Here we pass from $\qnc$ to nonlocal games to make an argument that works for general $\qnc$.
First let us fix ideas about $\qnc$.

\begin{defn}[$\qnc$]
\label{def:qnc}
	An $n$-input, depth-$d$ $\qnc$ circuit $\mc C$ is a quantum circuit composed of $d$ layers of arbitrary 2-qubit gates, acting on an \emph{input register} of $n$ qubits and an \emph{ancilla register} of $m\in \poly(n)$ qubits initialized to $\ket{0^m}$.
	Via measurement of the entire output of $\mc C$ in the computational basis, the circuit $\mc C$ effects a randomized mapping from $n$ bits of input to $n+m\in \poly(n)$ bits of output.
	A \emph{$\qnc$ circuit with $v$ qubits of quantum advice}, has $v$ out of $m$ ancilla qubits initialized to a $v$-qubit state, not necessarily a product state.
	For general $v\in\poly(n)$, this is denoted by the class $\qnc/\mathsf{qpoly}$.
\end{defn}

We will show a reduction from $\qnc$ circuits to nonlocal channels.

\begin{defn} (Nonlocal channel)
\label{def:nonlocal}
	Let $n,k\geq 1$ and $m\geq 0$.
	An $(n,k,m)$ nonlocal channel is the randomized mapping defined by a quantum strategy in a nonlocal game where $n$ parties receive one bit of input each and respond with $k$ bits each, along with a referee response of $m$ bits.
	
	Concretely, each party $i\in[n]$ is assigned a local Hilbert space $\mc H_i$ and for each $b\in\{0,1\}^n$, a POVM
	\[M_{(i,b)}=\big\{M_{(i,b)}^y:y\in\{0,1\}^k\big\}\]
	on $\mc H_i$.
	There is also a referee Hilbert space $\mc H_\mathrm{ref}$ with a fixed POVM
	\[M_\mathrm{ref}=\big\{M_\mathrm{ref}^y:y\in\{0,1\}^m\big\}\,.\]
	The definition of the nonlocal channel is completed by a choice of shared state $\ket{\psi}\in\big(\bigotimes_{i=1}^n\mc H_i\big)\otimes\mc H_\mathrm{ref}$ and works as follows.
	Upon receipt of an input string $x\in\{0,1\}^n$, the $n$ players and one referee perform the joint measurement $\big(M_{(1,x_1)},\ldots, M_{(n,x_n)},M_\mathrm{ref}\big)$ on $\ket{\psi}$, resulting in the outcomes $y_1,\ldots, y_n$, and $y_\mathrm{ref}$.
	The output of the channel is the $(nk+m)$-bit string $y=y_1||\cdots||y_n||y_\mathrm{ref}$.
\end{defn}
\begin{defn}[No-signaling channel]
	An $(n,k,m)$ no-signaling channel is defined analogously, except the correlations among parties may be general no-signaling correlations.
	(A very detailed definition of such channels is given in Definition \ref{defn:no-signaling}.)
\end{defn}

\begin{defn}[Parity games]
	Let $n,k,m$ be fixed and consider $f:\{0,1\}^{kn+m}\to \{0,1\}$.
	The $(n,k,f)$ parity game is played by $n$ entangled and non-communicating players, with the $i$\textsuperscript{th} player receiving input bit $x_i$ from $x$ drawn uniformly from $\{0,1\}^n$.
	A (quantum) parity game \emph{strategy} is an $(n,k,m)$ nonlocal channel with output string $y$.
	Players win when $f(y)=\text{\textsc{Par}}(x)$.
	We say a parity game strategy has \emph{advantage} $\epsilon$ if its winning probability is at least $1/2+\epsilon$.
\end{defn}

As a final piece of notation, for Boolean $f$ let $\neg f$ denote its negation.
We are prepared to give our reduction to parity games.

\begin{lemma}
	\label{lem:to-nlgs}
	Fix $n \geq 1, m\in \poly(n)$, let $\mc C$ be a $n$-qubit, depth-$d$ $\qnc$ circuit with $m$ ancilla and arbitrary quantum advice, and let $f:\{0,1\}^{n+m}\to\{0,1\}$ be any Boolean function.
	Suppose $f\circ \mc C$ has correlation $\epsilon$ with $\parity_n$.
	Then for some $n'\geq n/(2^d+1)$ there is a quantum strategy for the $(n',2^d, f)$ or $(n',2^d, \neg f)$ parity game with advantage $\epsilon/2$.
\end{lemma}
\begin{proof}
	Suppose $f\circ \mc C$ has correlation $\epsilon$ with $\parity$.
	For each input qubit $j$ denote by $L_j$ the set of output qubits in the forward lightcone of $j$.
	Consider the graph with vertices the input qubits $[n]$ and edges drawn between qubits $j$ and $k$ when $L_j$ and $L_k$ have nonempty intersection.
	Then $G$ has degree at most $2^d$, so there exists an independent set $S\subseteq[n]$ of size at least $n/(2^d+1)$.
	
	For each $y\in\{0,1\}^{S^c}$, define the circuit $\mc C_y$ to be $\mc C$ but where for $j\in S^c$, the $j$\textsuperscript{th} input is hardcoded to $y_j$.
	Then $\mc C_y$ is a circuit on at least $n/(2^d+1)$ variables such that the forward lightcones of input qubits are pairwise disjoint.
	Such a circuit defines an $(n', 2^d, m')$ nonlocal channel for some $n'\geq 2^{-d}+1$ and $m'=n+m-n'2^d$.
	(Note this $m'$ is without loss of generality because we may freely assign a player some output bits of the referee if their lightcone is smaller than $2^d$.)
	
	As a result, this restriction represents a strategy for the $(n',2^d, f)$ parity game.
	Moreover, we have
	\begin{align*}
		\E_x[(f\circ\mc C)(x)\cdot \parity(x)] &= \E_{y\sim\{0,1\}^{S^c}}\E_{z\sim\{0,1\}^{S}}[f\circ\mc C_y(z)\cdot \text{\textsc{Par}}(y||z)]\\
		&= \E_{y\sim\{0,1\}^{S^c}}\parity(y)\E_{z\sim\{0,1\}^{S}}[f\circ\mc C_y(z)\cdot \text{\textsc{Par}}(z)]\,.
	\end{align*}
	Therefore since $f\circ\mc C$ has $\epsilon$ correlation with parity on $n$ bits, for at least one $y$, $f\circ \mc C_y$ or $\neg f\circ \mc C_y$  must have at least $\epsilon$ correlation (in magnitude) with parity on $n/d$ bits.
	This is exactly half the advantage of the strategy defined by $\mc C_y$.
\end{proof}

Lemma \ref{lem:to-nlgs}
shows that bounds on the value of parity games translate into correlation bounds for $\acqnc$ with $\parity$.
How might we analyze parity games?
They are in some sense ``flipped'' versions of XOR games, where parity is computed on the inputs to the players, rather than the outputs.
However, it is not clear whether the rich collection of techniques developed to analyze XOR games is applicable here.
Instead, we bound the no-signaling value of the game by taking the perspective of distinguishability.

For any $(n,k,0)$ no-signaling channel $\mc N$, begin by rewriting the correlation as
\begin{align*}
	\E[(f\circ\mc N)(x)\cdot \parity(x)]=\frac{\E[(f\circ\mc N)(x)\mid x\text{ even}]-\E[(f\circ\mc N)(x)\mid x\text{ odd}]}{2}\,.
\end{align*}
Let $\mc U_\mathrm{even}$ and $\mc U_\mathrm{odd}$ denote the uniform distribution on even and odd bitstrings of length $n$ respectively, and consider the pushforwards of $\mc U_\mathrm{even}$ and $\mc U_\mathrm{odd}$ through $\mc N$:
\[\mu := \mc N\big(\mc U_\mathrm{even}\big)\qquad \text{and} \qquad \nu := \mc N\big(\mc U_\mathrm{odd}\big)\,.\]
So $\mu$ and $\nu$ are distributions on strings of length $N:=nk$, and
\begin{align*}
	\E[(f\circ\mc N)(x)\cdot \parity(x)]=\frac{\E[f(\mu)]-\E[f(\nu)]}{2}=\Pr[f(\mu)=1]-\Pr[f(\nu)=1]\,.
\end{align*}
Therefore the correlation of $f\circ \mc N$ with parity can be phrased in terms of $f$'s ability to distinguish the distributions $\mu$ and $\nu$.

What can be said about $\mu$ and $\nu$?
We claim that on every set $S\subset[N]$ of size at most $N/k-1=n-1$, we must have
\begin{equation}
\label{eq:mu-nu}
	\mu_S=\nu_S.
\end{equation}
Here the notation $\mu_S$ denotes the marginal distribution of $\mu$ on the coordinates in $S$.
To see \eqref{eq:mu-nu}, let $T\subset[n]$ be the set of players whose outputs overlap $S$.
Then by the no-signaling property of $\mc N$, the marginal $\mu_S$ (resp. $\nu_S$) is entirely determined by the marginal input distribution on $T$; that is, $(\mc U_\mathrm{even})_T$ (resp. $(\mc U_\mathrm{odd})_T$).
And for any $T$ a strict subset of $[n]$, $\mc (U_\mathrm{even})_T=(\mc U_\mathrm{odd})_T=\mc U(\{0,1\}^{|T|})$, so we must have $\mu_S=\nu_S$.

So all small marginals of $\mu$ and $\nu$ are information-theoretically indistinguishable.
This is exactly $k$-wise indistinguishability, a generalization of $k$-wise independence introduced by Bogdanov et al. \cite{bogdanov_bounded_2016} and first used in the context of secret sharing.

\begin{defn}[$k$-wise indistinguishability \cite{bogdanov_bounded_2016}]
	Two distributions $\mu$ and $\nu$ on $\{\pm 1\}^N$ are $k$-wise indistinguishable if for all $S\subset[N]$ with $|S|\leq k$, $\mu_S=\nu_S$.
	
	Additionally, for $f:\{0,1\}^n\to\{0,1\}$, we say $f$ is $\epsilon$-fooled by $k$-wise indistinguishability if for any pair $\mu, \nu$ of $k$-wise indistinguishable distributions,
\[|\Pr[f(\mu)=1]-\Pr[f(\nu)=1]|\leq \epsilon\,.\]
\end{defn}

It turns out $k$-wise indistinguishability over the hypercube is intimately connected to approximate degree.
By a linear programming duality argument, Bogdanov et al. proved the following.

\begin{theorem}[{\cite[Theorem 1.2]{bogdanov_bounded_2016}}]
	\label{thm:theapproxdeg}
	Let $f:\{0,1\}^n\to\{0,1\}$ and $\epsilon>0$.
	Then $f$ is $\epsilon$-fooled by $k$-wise indistinguishability if and only if $\widetilde{\deg}_{\epsilon/2}[f] \leq k\,.$
\end{theorem}

With this fact, Lemma \ref{lem:to-nlgs}, and the above discussion, we are ready prove the main theorem in this section.

We say a class of Boolean functions $\mc F=(\mc F_n)_{n\geq 1}$ is \emph{closed under inverse-polynomial restrictions} if for all $f\in \mc F_n$ and all $S\subseteq[n]$ with $n\in \poly(|S|)$, fixing the bits in $S^c$ yields a function still in $\mc F$:
\[f\restr{S^c\gets x}\in \mc F_{|S|}\quad \forall x\in \{0,1\}^{|S^c|}\,.\]
Note that $\ac$ is closed under inverse-polynomial restrictions.

\begin{theorem}
\label{thm:approx-deg}
	Suppose $\mc F$ is a class of Boolean functions closed under negations and inverse-polynomial restrictions.
	Let $m$ be fixed and suppose there is an $f\in \mc F$ on $N=\poly(m)$ variables and an $m$-input $\qnc$ circuit $\mc C$ of depth $d$, with $N-m$ ancilla qubits, and receiving arbitrary quantum advice, such that $f\circ\mc C$ achieves correlation $\epsilon$ with $\parity_m$.
	Then there is a $g\in \mc F$ on $n\geq m/2$ variables with $\widetilde{\deg}_{\epsilon/2}[g]\geq n/2^d-1$.
\end{theorem}
\begin{proof}
	By Lemma \ref{lem:to-nlgs}, there is an $m'\geq m/(2^d+1)$ and an $(m',2^d, N-2^dm')$ nonlocal channel $\mc N$ such that $f\circ \mc N$ or $\neg f\circ \mc N$ achieves correlation $\epsilon$ with $\parity_{m'}$.
	
	Suppose the referee measures their system and obtains outcome string $r$.
	This event leads to an updated state shared among the parties in $\mc N$ and thereby defines an $(m',2^d,0)$ nonlocal channel $\mc N^{R\gets r}$.
	By a similar averaging argument to the one used in Lemma \ref{lem:to-nlgs}, there is at least one outcome $r$ of the referee register such that $\mc N^{R\gets r}$ still yields correlation $\epsilon$ with $\parity$.
	Define $g:=f\restr{R\gets r}$ or $g:=\neg f\restr{R\gets r}$ as appropriate and put $\mc E:=\mc N^{R\gets r}$.
	Then $g\in \mc F$ is a function on $n:= 2^dm'$ bits and
	\[\E_x[(g\circ\mc E)(x)\cdot\parity(x)]\geq \epsilon\,.\]
	Therefore, by the discussion above, we see $g$ can $\epsilon$-distinguish $(n/2^d-1)$-wise indistinguishable distributions.
	Applying Theorem \ref{thm:theapproxdeg} we conclude that
	\[\widetilde{\deg}_{\epsilon/2}[g]\geq \frac{n}{2^d}-1\,.\qedhere\]
	\end{proof}
	\begin{corollary}
	\label{cor:deg}
		Suppose function class $\mc F$ is closed under inverse-polynomial-sized restrictions.
		Then if $\widetilde{\deg}[\mc F]\in o(n)$, $\mc F\circ\qnc$ cannot achieve $\Omega(1)$ correlation with $\parity$.
	\end{corollary}

	The burning question, then, is whether $\widetilde{\deg}[\ac]\in o(n)$.
	In fact, the approximate degree of $\ac$ is a longstanding open problem and its resolution would lead to several consequences in complexity theory \cite{bun_approximate_2022}.
	To get a sense of the difficulty of this question, consider that on one hand, a sublinear upper bound is known for a large subclass of $\ac$.
	\begin{theorem*}[{\cite[Theorem 5]{bun_quantum_2019}}]
		Let $p(n)\in \poly(n)$.
		Then the class of $\ac$ circuits of linear size, denoted by $\mathsf{LC^0}$, has
		\[\widetilde{\deg}_{1/p(n)}[\mathsf{LC^0}]\in o(n).\]
	\end{theorem*}
	Yet on the other hand, a series of works, most recently \cite{sherstov_approximate_2022}, show the following:
	\begin{theorem*}
		For any $\delta >0$, there is a function $f\in \ac$ with $\widetilde{\deg}[f]\in\Omega(n^{1-\delta})$.
	\end{theorem*}
	
	The lower bound of $\Omega(n^{1-\delta})$-for-any-$\delta$ is tantalizingly close to the trivial upper bound of $n$ for the approximate degree of any Boolean function, but as it stands it is not unreasonable to guess that $\widetilde{\deg}[\ac]\in\Theta(n/\log n)$ either.
	Several questions---including now Conjecture \ref{conj:main}---could be settled if the gap between $\Omega(n^{1-\delta})$-for-any-$\delta$ and $n$ for $\widetilde{\deg}[\ac]$ were closed.
	
	We may combine the sublinear lower bound on $\mathsf{LC^0}$ from \cite{bun_quantum_2019} with Theorem \ref{thm:approx-deg} to obtain:
	\begin{corollary}
	\label{cor:qnc0lin}
		Let $\mc C$ be an $n$-input, $m$-ancilla $\qnc$ circuit with arbitrary advice.
		Suppose $f:\{0,1\}^{n+m}\to\{-1,1\}$ is defined by an $\ac$ circuit of size $\mathcal{O}(n)$.
		Then $f\circ\mc C$ achieves negligible correlation with $\parity_n$.
	\end{corollary}

	\subsection{Blockwise approximate degree}
	We conclude this section by laying out a self-contained question concerning the approximate degree of $\ac$ with respect to a modified, ``blockwise'' notion of approximate degree.
	This question is sufficient to imply Conjecture \ref{conj:main} in full generality and may be easier to resolve than $\widetilde{\deg}[\ac]$.
	
	Fix $k\geq 1$ (assuming $k$ divides $n$ for simplicity) and let $P$ be the partition of $[n]$ into ``blocks'' of size $k$:
	\[P:=\big\{\{1,\ldots, k\},\{k+1,\ldots,2k\},\ldots,\{n-k+1,n\}\big\}\,.\]
For a monomial $\bigchi_S=\prod_{i\in S}x_i$ define the ($k$-)\emph{block degree} $\mathrm{bdeg}_k[\bigchi_S]$ to be the number of distinct blocks $B\in P$ having nonempty intersection with $S$.
This definition extends naturally to the $k$-block degree $\mathrm{bdeg}_k[f]$ of a Boolean function $f:\{0,1\}^n\to\{-1,1\}$ and to the \emph{approximate $k$-block degree} $\widetilde{\mathrm{bdeg}}_k[f]$ of $f$:
\[\widetilde{\mathrm{bdeg}}_k[f]=\min\{\mathrm{bdeg}_k[g]\mid g:\{0,1\}^n\to\mathbb{R}\text{ a polynomial with } \|f-g\|_\infty\leq 1/3\}\,.\]
Of course $\widetilde{\mathrm{bdeg}}_k[f]\leq n/k$ for any function.

\begin{question}
\label{q:blockwise}
	For all constants $k$, does the following hold?
	\[\widetilde{\mathrm{bdeg}}_k[\mathsf{AC^0}]\,\overset{?}{\leq}\, n/k-1\,.\]
\end{question}

As we explain below, this would be enough to prove Conjecture \ref{conj:main}.
Note the following, which are immediate and hold for all $f$:
\[\widetilde{\deg}[f]< \frac{n}{k} \;\implies\; \widetilde{\mathrm{bdeg}}_k[f] < \frac{n}{k}\;\implies\; \widetilde{\deg}[f]< n-k\,.\]
Moreover, these implications are sharp in that each one cannot generically imply anything stronger, as witnessed by a parity function on an appropriate subset of $[n]$.
Regarding $f\in \mathsf{AC^0}$, the left-hand side holding for arbitrary constant $k$ is equivalent to $\widetilde{\deg}[\mathsf{AC^0}]\in o(n)$, while the far right-hand side follows directly from LMN-type Fourier tail bounds for $\mathsf{AC^0}$.

\begin{prop}
	If the resolution to Question \ref{q:blockwise} is ``yes'', then Conjecture \ref{conj:main} is true.
\end{prop}
\begin{proof}[Proof sketch.]
	Consider the referee-free nonlocal channel $\mc E$ from the proof of Theorem \ref{thm:approx-deg}, with $n/k$ players responding with $k$ bits each.
	Defining $\mu$ and $\nu$ as the pushforwards of uniform distributions over even and odd bitstrings as before, it is true that $\mu$ and $\nu$ are $(n/k-1)$-wise indistinguishable when viewed as distributions on $\{0,1\}^n$.
	However, they may also be viewed as distributions on the hypergrid $[2^k]^{m}$ for $m=n/k$.
	
	With this perspective, $\mu$ and $\nu$ are $m-1$ indistinguishable.
	Repeating the proof of {\cite[Theorem 1.2]{bogdanov_bounded_2016}} over this larger alphabet, we recover exactly the notion of blockwise degree.
	The rest of the argument is as before.
\end{proof}

It is unclear to us whether Question \ref{q:blockwise} is easier than $\widetilde{\deg}[\ac]\overset{?}{\in}o(n)$.
Because $\ac$ is closed under permutations of input coordinates $[n]$, we can compare the two questions head-to-head as follows.
Let $\mc P_k$ be all the relabelings of $P$:
\[\mc P_k :=\Big\{\big\{\{\pi(1),\ldots, \pi(k)\},\{\pi(k+1),\ldots, \pi(2k)\},\ldots, \{\pi(n-k+1),\ldots, \pi(n)\}\big\}\Big\}_{\pi\in S_n}\,.\]
For any $P\in \mc P_k$, let $\mathrm{bdeg}_P[f]$ be the maximum number of blocks in $P$ overlapped by some monomial in $f$.
Then we have the following characterization, where $g$ ranges over real-valued multilinear polynomials on the hypercube as usual:
\begin{align*}
\widetilde{\deg}[\ac]< n/k &\iff \forall f\in \ac,  \exists g, \forall P\in \mc P_k, \mathrm{bdeg}_P[g]\leq n/k \text{ and } \|f-g\|_\infty\leq 1/3\\
\widetilde{\mathrm{bdeg}}_k[\ac]<n/k &\iff \forall f\in \ac, \forall P\in \mc P_k, \exists g, \mathrm{bdeg}_P[g]\leq n/k \text{ and } \|f-g\|_\infty\leq 1/3\,.
\end{align*}

\section{Towards a switching lemma for $\acqnc$}
\label{sec:switching-lemma}

Recall that our approach in Section \ref{sec:ancilla-free} fails because circuits with LMN-style Fourier decay are not suitably closed under precomposition by $\qnc$.
In fact this is true even under precomposition by $\nc$, and the proof of the LMN theorem elegantly avoids an induction assumption phrased in terms of Fourier decay.
Instead, the proof relies on a structural theorem about the effect of random restrictions on DNFs and CNFs---H\aa stad's celebrated switching lemma:

\begin{theorem*}[H\aa stad \cite{hastad_almost_1986}]
	Suppose $f$ is a width-$w$ DNF.
	Then for any $0\leq \delta \leq 1$,
	\[\Pr_{\rho\sim \mathbf{R}_\delta}[\mathrm{DT}_\mathrm{depth}(f\restr{\rho})>t]\leq (C\delta w)^t\,,\]
	where $C$ is a universal constant.
\end{theorem*}

Here $\mathbf{R}_\delta$ is the distribution of random restrictions with star probability $\delta$ (see \emph{e.g.,} \cite[\S 4.3]{ODonnell2014-li} for more).
This theorem has received several proofs over time, but each rely on the well-controlled structure of random restrictions.
To naively repeat the switching lemma argument directly on $\acqnc$ would mean to track the passage of random restrictions through $\qnc$---a tall order given that $\qnc$ can destroy the independence and unbiasedness of random restrictions that switching arguments tend to rely on.

The situation may be slightly improved by instead considering a switching lemma for the model studied in Section \ref{sec:nlgs}.
Recalling that $f\circ \mc N$ is a randomized function, we may hope for a switching lemma of the following form:

\begin{quote}
\noindent\emph{An imagined switching lemma for nonlocal channels}.
Let $m\geq 0$ and $k,w,n\geq 1$ and suppose $f:\{0,1\}^{kn+m}\to\{0,1\}$ is a DNF of width $w$ and $\mc N$ is an $(n,k,m)$ nonlocal channel.
Then for each restriction $\rho$ there exists a distribution $\Gamma\!_\rho$ over decision trees such that $(f\circ\mc N)\restr{\rho}=\{T\}_{T\sim \Gamma\!_\rho}$ and
\[\Pr_{\rho\sim \mathbf{R}_\delta}\Pr_{T\sim \Gamma_\rho}[\mathrm{depth}(T)>t]\leq (C\delta w)^t\,.\]
\end{quote}
By Lemma \ref{lem:to-nlgs} such a switching lemma would be sufficient to show correlations bounds between $f \circ \qnc$ and parity for any DNF (or CNF) $f$, which in turn are direct prerequisites to proving Conjecture \ref{conj:main}.
While this imagined switching lemma is currently out of reach, we contend it presents a useful challenge to existing switching lemma proof techniques.
As a first step in this direction, we devote this section to a proof of a simpler but related structural result.

\medskip
\noindent\textbf{Theorem.} (Informal)
\emph{Any no-signaling channel $\mc N$ composed with a decision tree $\tau$ is equal to a probability distribution $\Gamma$ of decision trees with $\mathrm{depth}(\tau')\leq \mathrm{depth}(\tau)$ for all $\tau'\in\mathrm{Supp}(\Gamma)$.}
\medskip

Let us fix some notation.
For a finite set $X$ let $\mc M(X)$ denote the set of probability measures on $X$.
The set $\mc M(X)$ is convex, so for $\nu$ a probability measure on $\mc M(X)$ we may define the expected distribution
\begin{equation}
\label{def:avgdist}
	\E_{\mu \sim \nu}[\mu] := \left\{x\hspace{.5em} \text{w.p.}\hspace{.5em} \Pr_{\mu \sim \nu}\Pr_{z\sim \mu}[z = x]\right\}_{x\in X}
\end{equation}
Here we study Boolean channels, or functions of the form
\[\mc N:\{\pm 1\}^n \to \mc M(\{\pm1\}^{N}).\]
For a probability measure $\mu$ on the set of channels from $n$ to $N$ bits, we use $\E_{\mc N \sim \mu}\mc N$ to denote the channel defined pointwise as:
\begin{equation}
\label{eq:avg-channel}
\Big(\E_{\mc N \sim \mu}\mc N\Big)(x) := \E_{\mc N \sim \mu}[\mc N(x)].
\end{equation}
To be clear,  $\mc N(x)$ is a probability measure on $\{\pm 1\}^N$, so in the right-hand side of \eqref{eq:avg-channel} we are computing the expected distribution according to \eqref{def:avgdist}.
Also, for $T\subseteq[N]$ define the \emph{reduced channel}
\begin{align*}
    \mc N^T(x) := \Big\{y\hspace{.5em}\text{ w.p.} \sum_{\substack{z\in\{\pm 1\}^N\\z_T = y}} \Pr[\mc N(x) = z]\Big\}_{y\in\{\pm 1\}^{|T|}}.
\end{align*}

\begin{defn}[No-signaling channel]
\label{defn:no-signaling}
Consider a map $\mc N:\{\pm 1\}^n \to \mc M(\{\pm1\}^{N})$ and a `backwards lightcone' function $B:[N]\to[n]\cup\{\perp\}$. 
The pair $(\mc N, B)$ is a \emph{no-signaling channel} (NSC) if for all $S\subseteq[n]$, for all $x,x'\in\{\pm1\}^n$ with $x_S = x'_S$, we have $\mc N^{B^{-1}(S\cup{\perp})}(x) = \mc N^{B^{-1}(S\cup{\perp})}(x')$.
\end{defn}

\noindent That is, a channel is an NSC if for any collection of output indices $T$, $\mc N^{T}(x)$ is a function of $x_{B(T)\backslash\{\perp\}}$ only.
Note also $\mc N^{B^{-1}(\perp)}$ is oblivious to the value of $x$ entirely---the outputs $B^{-1}(\perp)$ could be called the referee outputs.

Recall that for a Boolean function $f:\{\pm 1\}^N\to \{\pm 1\}$, $f\circ \mc N$ denotes the channel
\[f\circ \mc N(x) = \big\{b\hspace{1em}\text{ w.p.}\hspace{.5em}\Pr_{y\sim \mc N(x)}[f(y) = b]\big\}_{b\in\{\pm 1\}}.\]
The restriction structure on NSCs interacts nicely with decision trees:

\begin{theorem}
\label{thm:tree-decomp}
Given $f:\{\pm1\}^{N}\to\{\pm 1\}$ and $\mc N: \{\pm 1\}^n\to\mc M(\{\pm 1\}^{N})$ an NSC, there exists a distribution $\Gamma$ over decision trees such that
\begin{enumerate}[label=\roman*.]
    \item For all $x$ the composition $f\circ\mc N (x) = \{\tau(x)\}_{\tau\sim \Gamma}$, so $\E[f\circ \mc N(x)] = \E_{\tau\sim \Gamma}[\tau(x)]$; and
    \item For all $\tau\in \mathrm{Supp}(\Gamma)$, $\mathrm{DT}_\text{depth}(\tau)\leq \mathrm{DT}_\text{depth}(f)$.
\end{enumerate}
\end{theorem}
Recall that $f\circ \mc N$ is an $\mc M(\{\pm1\})$-valued function on the hypercube, so $x\mapsto\E[f\circ\mc N(x)]$ is a $[-1,1]$-valued function on the hypercube, and accordingly has a multilinear Fourier expansion
\[\E[f\circ\mc N]=\sum_{S\subset[n]}a_S\bigchi_S\qquad\text{with}\qquad a_S:=\E_x\big[\E[(f\circ\mc N)(x)]\cdot\bigchi_S(x)]\]
We pause to note the related fact that in terms of the expected output $\E[f\circ\mc N]$, the degree of any function $f$ does not increase under composition with an NSC: $\deg(f)\geq\deg( \E[f\circ\mc N])$.
This claim has a very simple direct proof\footnote{Consider the Fourier expansion $f=\sum_{S\subseteq[N]}\widehat{f}(S)\bigchi_S$.
Then $\E[(f\circ\mc N)(x)] = \E\big[\sum_{S\subseteq[N]}\widehat{f}(S)\bigchi_S\circ\mc N(x)\big] = \sum_S\widehat{f}(S)\E[\bigchi_S\circ\mc N(x)] = \sum_S\widehat{f}(S)\E[\bigchi_S\circ\mc N^S(x)]$, a linear combination of functions of at most $|S|$ variables each for $|S|\leq\deg(f)$.} and we emphasize that it is not equivalent to Theorem \ref{thm:tree-decomp}.
For example, there are Boolean functions $g$ with $\deg(g) = n^{2/3}$ but $\mathrm{DT}_\text{depth}(g) = n$ (see Example 3 in \cite{buhrman_complexity_2002}).
One could imagine a Boolean function $h$ with $\deg(h)\approx\mathrm{DT}_\text{depth}(h)\in o(n)$ but where $\E[h\circ \mc N]$ is ``$g$-like'': any decision tree decomposition of $\E[h\circ \mc N]$ contains a tree of depth $n$ despite having $\deg(\E[h\circ\mc N])\in o(n)$.
Theorem \ref{thm:tree-decomp} says such an $h, \mc N$ pair does not exist; precomposition by an NSC cannot increase the decision tree complexity of a function.

The proof of Theorem \ref{thm:tree-decomp} requiries some bookkeeping.
The idea is to begin with $\tau$'s root vertex variable $y_i$ and locally decompose the univariate channel ${\mc N}^{\{i\}}(x)\mapstochar\mathrel{\mspace{1.95mu}}\leadsto y_i$ into a distribution of deterministic functions $\{y_{i,\omega}(x_i)\}_{\omega\sim\mu}$.
This decomposition of the root vertex induces a probabilistic decomposition $\{\tau'_\omega\circ \mc N'_\omega\}_{\omega\sim \mu}$ of the entire hybrid computation where the root variable $y_i$ in $\tau'$ has been replaced with an $x_{B(i)}$ and the left 
and right subtrees of $\tau$ become compositions not with $\mc N$, but with conditional versions of $\mc N$ where $x_{B(i)}$ and $y_i$ have been fixed to certain values.
This conditioning preserves the NSC-ness of the new $\mc N'$s, and the decomposition recurses down the tree.

We now introduce a notion of conditioning.
For any $n$-to-$N$ bit Boolean channel $\mc N$, $x\in \{\pm 1\}^n$, $J\subseteq[N]$ and $Y\in\mathrm{Supp}(\mc N^J(x))$ define the \emph{conditional channel} as
\[\mc N(x\mid y_J = Y) := \big\{y\hspace{.5em} \text{ w.p. } \Pr[\mc N(x) = y \mid y_J = Y]\big\}_{y\in\{\pm 1\}^N},\]
and for $T\subseteq[N]$ the \emph{reduced conditional channel}
\begin{align*}
    \mc N^T(x \mid y_J = Y) := \Big\{y\hspace{.5em}\text{ w.p.} \sum_{\substack{z\in\{\pm 1\}^N\\z_T = y}} \Pr[\mc N(x) = z\mid z_J = Y]\Big\}_{y\in\{\pm 1\}^{|T|}}.
\end{align*}

Note that $T$-reduced conditional no-signaling channels can depend on inputs outside $B(T)$.
Consider for example the $n$-to-$n$-bit NSC
\[\mc G(x) = \begin{cases}
\,\mc U\{\text{even strings}\} &x\text{ even}\\
\,\mc U\{\text{odd strings}\} &x\text{ odd}.
\end{cases}\]
Now $\mc G^{\{i\}}(x)$ is identically a Rademacher random variable (oblivious to $x$ entirely), but
\[\mc G^{\{i\}}(x\mid y_{[n]\backslash{i}} = 00\cdots 0) = \Big\{\prod_jx_j \quad \text{w.p. } 1\Big\},\]
the parity of all $n$ bits of $x$.
All the same, some structure remains after conditioning:

\begin{prop}
For $T,J\subseteq[N]$, let $x,x'\in\{\pm 1\}^n$ be such that $x_{B(J\cup T)} = x'_{B(J\cup T)}$.
Then for all $Y\in\mathrm{Supp}(\mc N^J(x))$,
\[\mc N^T(x\mid y_J = Y) \hspace{.5em} = \hspace{.5em} \mc N^T(x'\mid y_J = Y).\]
\end{prop}
\begin{proof}
Let $x,x'$ be as in the proposition statement.
We have from the definition of NSCs that $\mc N^{J\cup T}(x) = \mc N^{J\cup T}(x').$
Certainly then
$\mc N^{J\cup T}(x\mid y_J = Y) = \mc N^{J\cup T}(x'\mid y_J = Y)$ (we have taken the marginal of two equal distributions).
The conclusion then follows from noticing that for any $U\subseteq V$,
$\mc N^U = (\mc N^V)^U$.
\end{proof}
This proposition says $\mc N^{T}(x\mid y_J = Y)$ is a function of $x_{B(J\cup T)}$ only.
Thus if we fix variables $x_{B(J)}$ we recover a smaller NSC:
\begin{corollary}
\label{cor:restrictions-are-fine}
Consider an $n$-to-$N$ NSC $(\mc N, B)$, an $i\in [N]$, and $X,Y\in \{\pm 1\}$.
If $B(i) = \perp$ let $\mc N'$ be the $n$-to-$(N-1)$ NSC
\[\mc N' = \mc N^{[N]\backslash \{i\}}(x \mid y_i = Y)\]
and otherwise let $\mc N'$ be the $(n-1)$-to-$(N-1)$ NSC
\[\mc N' = \mc N^{[N]\backslash \{i\}}(x_{\{B(i)\}^c} \mid x_{B(i)} = X, y_i = Y).\]
Define a new lightcone function $B'$ from $B$ as follows.
Put $B(j) = \perp$ for all $j\in B^{-1}(B(i))$ and then remove $i$ from the domain of $B$.
Then $(\mc N', B')$ is an NSC.
\end{corollary}

Finally we introduce an object used internally in the proof of Theorem \ref{thm:tree-decomp}.

\begin{defn}[Hybrid Decision Tree]
A \emph{hybrid decision tree} $\mc T$ on $n$ variables with $\ell$ leaves consists of the data $(\tau, \mc G_1,\ldots, \mc G_\ell)$, where
\begin{enumerate}[label=\roman*.]
\item The first argument $\tau$ is a rooted binary tree with $\ell$ leaves labeled as follows. Each internal node is assigned $x_i$ for some $i\in [n]$, the edge to its left child is labeled $1$, and the edge to its right child is labeled $-1$.
\item Each leaf $\iota$ of $\tau$ is associated with an $n$-to-1 channel $\mc G_\iota:\{\pm 1\}^n\to \mc M\{\pm 1\}$.
\end{enumerate}
A hybrid tree defines a channel $\mc T_\tau(\mc G_1,\ldots,\mc G_\ell):\{\pm 1\}^n\to \mc M\{\pm 1\}$ as follows.
Computation on input $x\in \{\pm 1\}^n$ proceeds just as with standard decision trees until a leaf $\iota$ is reached, at which point the distribution $\mc G_\iota(x)$ is returned.
\end{defn}

Theorem \ref{thm:tree-decomp} follows from these three claims.
Proofs of the first two are immediate from the definitions.

\begin{claim}
\label{clm: expectation-unwrap}
For any hybrid decision tree $\mc T$,
\[\mc T\big(\mc G_1,\ldots,\mc G_{\iota-1},\E_{\omega\sim \mu}[\mc G_\omega], G_{\iota+1},\ldots, \mc G_\ell\big) = \E_{\omega\sim \mu}\big[\mc T(\mc G_1,\ldots,\mc G_{\iota-1},\mc G_\omega, \mc G_{\iota+1},\ldots, \mc G_\ell)\big]\]
\end{claim}

\begin{claim}
\label{clm: hybrid-tree-comp}
For any hybrid decision trees $\mc T_\tau(\mc G_1,\ldots,\mc G_\ell)$ and $\mc T_{\tau'}(\mc G_{\iota1},\ldots, \mc G_{\iota \ell'})$,
\begin{align*}
&\mc T_\tau\big(\mc G_1,\ldots,\mc G_{\iota-1},\mc T_{\tau'}(\mc G_{\iota 1},\ldots, \mc G_{\iota \ell'}), G_{\iota+1},\ldots, \mc G_\ell\big)\\
&\hspace{8em} = \mc T_{\tau\circ_\iota \tau'}(\mc G_1,\ldots,\mc G_{\iota-1},\mc G_{\iota1},\ldots,\mc G_{\iota \ell'}, \mc G_{\iota+1},\ldots, \mc G_\ell),
\end{align*}
where $\tau\circ_\iota\tau'$ is $\tau$ with the $\iota$\textsuperscript{th} leaf replaced with $\tau'$.
\end{claim}

\begin{claim}
\label{clm: vertex-decomp}
Suppose $\tau$ is a decision tree and $(\mc N, B)$ is an NSC.
Then either:
\begin{enumerate}[label=\roman*., font=\itshape]
    \item $\tau\circ\mc N = \E_{\omega\sim\mu}[\tau_\omega\circ \mc N_\omega]$ where $\mathrm{depth}(\tau_\omega)\leq \text{depth}(\tau)-1$, $|\mathrm{Supp}(\mu)|\leq 2$, and each $\mc N_\omega$ is an NSC, or
    \item $\tau\circ\mc N = \E_{\omega\sim \mu}\Big[\mc T_{\tau^*}\big(\tau_{\omega_L}\circ \mc N_{\omega_L}, \tau_{\omega_R}\circ \mc N_{\omega_R}\big)\Big]$,
    where $|\mathrm{Supp}(\mu)|\leq 3$, $\tau^*$ has one internal node, $\mathrm{depth}(\tau_{\omega_L}),\mathrm{depth}(\tau_{\omega_R})\leq \mathrm{depth}(\tau)-1$, and each $\mc N_{\omega_L}, \mc N_{\omega_R}$ is an NSC; or
    \item (Base case) $\tau\circ \mc N(x) = \{b \,\,\text{ w.p. } 1\}$ for all $x$, for some fixed $b\in \{\pm 1\}$.
\end{enumerate}
\end{claim}
\begin{proof}
If $\tau$ is the trivial decision tree with no internal nodes, clearly we satisfy case \emph{iii}.
Otherwise, let $y_i$ be the variable at the root of $\tau$.
There are two cases depending on the value of $B(i)$.

\emph{Case i), $B(i)=\,\perp$.}
Observe that $\mc N^{\{i\}}(x)$ is the same distribution $\mu$ over $\{\pm 1\}$, independent of $x$.
For $\omega \in \{\pm 1\}$ let $\tau_\omega$ be the subtree of $\tau$ attached to the $\omega$-valued edge of $y_i$.
Put $\mc N_\omega = \mc N^{T\backslash\{i\}}(x\mid y_i = \omega)$.
Then we have for $z\in\{\pm 1\}$,
\begin{align*}
    \Pr[\tau \circ \mc N(x) = z] &= \textstyle\sum_{\omega\in\{\pm1\}}\Pr[\tau\circ\mc N(x)=z\mid D^i(x) = \omega]\Pr[\mc N^i(x) = \omega]\\
    &= \textstyle\sum_{\omega\in\{\pm1\}}\Pr[\tau_\omega\circ\mc N(x\mid y_i = \omega) = z]\Pr[\mc N^i(x) = \omega]\\
    &= \textstyle\sum_{\omega\in\{\pm1\}}\Pr[\tau_\omega\circ\mc N_\omega(x) = z]\Pr[\mc N^i(x) = \omega]\\
    &= \Pr\big[\E_{\omega \sim \mu}[\tau_\omega\circ\mc N_\omega](x) = z\big]
\end{align*}
as desired.
Clearly $\tau_\omega$ is strictly shorter than $\tau$, and $\mc N_\omega$ is an NSC by Corollary \ref{cor:restrictions-are-fine}.

\emph{Case ii), $B(i)\neq \,\perp$.}
Let $\tau^*$ be the one-vertex tree consisting of the root vertex of $\tau$ relabeled with $x_{B(i)}$ and let $\tau_{1}$, $\tau_{-1}$ be the left and right subtrees of $\tau$ respectively.
Observe that $\mc N^{\{i\}}(x)=\mc N^{\{i\}}(x_{B(i)})$ is a univariate channel.
Hence it can be decomposed as a convex combination
\[\mc N^{\{i\}}(x_{B(i)}) = a_{(1,1)}\begin{bmatrix}
1 & 1\\
0 & 0
\end{bmatrix} + a_{(-1,-1)}\begin{bmatrix}
0 & 0\\
1 & 1
\end{bmatrix} + a_{(1,-1)}\begin{bmatrix}
1 & 0\\
0 & 1
\end{bmatrix} + a_{(-1,1)}\begin{bmatrix}
0 & 1\\
1 & 0
\end{bmatrix}.\]
where only three of $a_{(L,R)}$ are nonzero.
Let $\mu = \{(L,R)\text{ w.p. } a_{(L,R)}\}$.
Then we claim
\begin{equation}
    \tau \circ \mc N = \E_{(L,R)\sim \mu}\big[\mc T_{\tau^*}\big(\tau_{L}\circ \mc N^{(1)}_{L}, \tau_{R}\circ \mc N^{(-1)}_{R}\big)\big]
    \label{eq:dt-nsg-case-2-eq}
\end{equation}
where for $b,c\in\{\pm1\}^2$,
\[\mc N_{c}^{(b)}(x) = \mc N(x|x_{B(i)} = b, y_i = c).\]
We check Eq. \eqref{eq:dt-nsg-case-2-eq} pointwise.
First consider an $x$ with $x_{B(i)} = 1$.
We condition on the value of $y_i$, rearrange, and then ``complete the tree'':
\begin{align*}
    \Pr[\tau\circ\mc N(x) = z] &= \sum_{L\in\{\pm 1\}}\Pr[\tau\circ\mc N(x)= z \mid \mc N_i(x) = L]\Pr[\mc N_i(x) = L]\\
    &= \sum_{L\in\{\pm 1\}}\Pr[\tau\circ\mc N(x\mid y_i = L)= z ](a_{(L, 1)} + a_{(L, -1)})\\
    &= \sum_{L\in\{\pm 1\}}\Pr[\tau_L\circ\mc N(x\mid x_{B(i)}=1, y_i = L)= z ]\big(\textstyle\sum_{R\in\{\pm 1\}}a_{(L, R)}\big)\\
    &= \sum_{L,R\in\{\pm 1\}}a_{(L,R)}\Pr[\tau_L\circ\mc N^{(1)}_{L}(x)= z ]\\
    &= \sum_{L,R\in\{\pm 1\}}a_{(L,R)}\Pr[\mc T_{\tau^*}(\tau_L\circ\mc N^{(1)}_{L},\tau_R\circ\mc N^{(-1)}_{R})(x)= z ]\\
    &= \Pr\Big[\E_{(L,R)\sim \mu}[\mc T_{\tau^*}(\tau_L\circ\mc N^{(1)}_{L},\tau_R\circ\mc N^{(-1)}_{R})](x) = z\Big],
\end{align*}
as desired.
A similar argument goes through for $x_{B(i)}=-1$ by expanding over $R$ instead of $L$.
\end{proof}

\begin{proof}[Proof of Theorem \ref{thm:tree-decomp}]
Let $\tau$ be a depth-optimal decision tree for $f$.
Construct the trivial hybrid tree $\mc T$ with no internal nodes and a single leaf with label $\tau\circ \mc N$.
Put $\Gamma = \{\mc T\text{ w.p. }1\}.$
We will recursively break apart leaves of $\mc T$ into distributions of hybrid trees, which are then combined with the parent tree to become distributions over hybrid trees of greater depth.

This is done by repeated application of the following sequence of steps.
Suppose $\mc T_\tau(\mc G_1,\ldots, \mc G_\ell)$ is some hybrid tree and $\mc G_\iota = \tau' \circ \mc N$ for some nontrivial DT $\tau'$ and (potentially conditioned) NSC $\mc N$.
Then depending on the case in Claim 3 we either have
\begin{align*}
\mc T_\tau(\ldots,\underbrace{\tau\circ \mc N}_{\text{index } \iota}, \ldots)
    &=\mc T_\tau(\ldots,\textstyle\E_{(L,R)\sim \mu} [\mc T_{\tau^*}(\tau_{\omega_L}\circ\mc N_{\omega_L},\tau_{\omega_R}\circ\mc N_{\omega_R})],\ldots)
    \tag{Claim \ref{clm: vertex-decomp}.\emph{i}}\\
    &=\E_{(\omega_L,\omega_R)\sim \mu}\big[\mc T_\tau(\ldots,\mc T_{\tau^*}(\tau_{\omega_L}\circ\mc N_{\omega_L},\tau_{\omega_R}\circ\mc N_{\omega_R}),\ldots)\big]
    \tag{Claim \ref{clm: expectation-unwrap}}\\
    &= \E_{(\omega_L,\omega_R)\sim \mu}\big[\mc T_{\tau\circ_\iota \tau^*}(\ldots,\tau_{\omega_L}\circ\mc N_{\omega_L},\tau_{\omega_R}\circ\mc N_{\omega_R},\ldots)\big]
    \tag{Claim \ref{clm: hybrid-tree-comp}},
\intertext{where $\tau^*$ has depth 1 and $\mathrm{depth}(\tau_{\omega_L}), \mathrm{depth}(\tau_{\omega_R})\leq \mathrm{depth}(\tau')-1$, or we have}
\mc T_\tau(\ldots,\underbrace{\tau\circ \mc N}_{\text{index } \iota}, \ldots)
    &=\mc T_\tau(\ldots,\textstyle\E_{\omega\sim \mu} [\tau_{\omega}\circ\mc N_{\omega}],\ldots)
    \tag{Claim \ref{clm: vertex-decomp}.\emph{ii}}\\
    &=\E_{\omega\sim \mu}\big[\mc T_\tau(\ldots,\tau_{\omega}\circ\mc N_{\omega},\ldots)\big]
    \tag{Claim \ref{clm: expectation-unwrap}},
\end{align*}
where $\mathrm{depth}(\tau_{\omega}) \leq \mathrm{depth}(\tau')-1$.

If we repeatedly make these transformations on the elements of $\Gamma$, we will eventually be left with a distribution over hybrid decision trees $(\tau, \mc G,\ldots)$ where each channel $\mc G = \tau'\circ \mc N$ is in the base case of Claim \ref{clm: vertex-decomp}.
Such a hybrid tree is equal to a deterministic channel.
Hence we are left with a distribution over deterministic channels that is trivially equivalent to a distribution of standard, deterministic decision trees.

Further, it's easy to see that once done, the longest path in any tree of $\mathrm{Supp}(\Gamma)$ is bounded by the longest path in the original tree $\tau$.
\end{proof}

\section{Discussion}
\label{sec:discussion}

We have seen several pieces of evidence for Conjecture \ref{conj:main}, as well as highlighted new connections between quantum complexity theory, nonlocal games, and approximate degree.

If Conjecture 1 is ultimately proved true, we may wish to reach for a stronger no-advantage theorem closer to that of Beals et al. \cite{beals_quantum_2001} from query complexity.
A natural expression of  $\acqnc$ non-advantage might use the language of Fourier decay.

\begin{question}
	Does $\acqnc$ exhibit LMN-like Fourier decay?
	To make this precise for the randomized function $f\circ \mc C$, consider the expectation over the randomness in $\mc C$ to get a function $F:\{0,1\}^n\to [-1,1]$.
	Then we ask, is $\mathbf{W}^{\geq t}[F]\in\mathcal{O}(\exp(-t))$?
\end{question}
\noindent As mentioned in the introduction, a similar result is known depth-$d$ $\qac$ circuits with at most $\mc O(n^{1/d})$ ancillas \cite{paulispec}.

Finally, one may consider any number of variations on the theme of precomposing a Boolean function with $\qnc$.
It is natural to ask:

\begin{question}
\label{q:map}
	View a $\qnc$ circuit $\mc C$ as a map from (randomized) Boolean functions to randomized Boolean functions:
	\[f\overset{\mc C}\longmapsto f\circ \mc C\,.\]
	By how much can this map increase influence, sensitivity, or other complexity measures of $f$?
\end{question}
\noindent Theorem \ref{thm:tree-decomp} gives the answer ``not at all'' to a variant Question \ref{q:map} where $\qnc$ is replaced by nonlocal channels, and the complexity measure is randomized decision tree complexity.

\section{Acknowledgements}
We are grateful to Chris Umans and Thomas Vidick for numerous valuable discussions and for the opportunity to share this work with Henry Yuen and the quantum group at Columbia University in the fall of 2022.
We are also grateful to Atul Singh Arora, discussions with whom inspired this project.
Finally we thank the anonymous ITCS 2024 reviewers for their generous and meticulous feedback on an earlier draft.

\section*{References}
{\def\section*#1{}
\bibliographystyle{plain}
\bibliography{main.bib}
}
\end{document}